\documentclass[submission,copyright,creativecommons]{eptcs}
\providecommand{\event}{FROM 2024} 

\usepackage[disable]{todonotes}
\usepackage{mathtools}
\usepackage{amsthm}
\usepackage{amssymb, amsmath}
\usepackage{multirow}
\usepackage{array}
\usepackage{setspace}
\usepackage{enumitem}
\usepackage[square, numbers]{natbib}
\usepackage{cleveref}
\usepackage{ebproof}
\usepackage{subcaption}

\usepackage{iftex}

\MakeRobust{\ref}

\makeatletter
\newcommand{\labeltext}[2]{%
  #1
  \@bsphack
  \csname phantomsection\endcsname 
  \def\@currentlabel{#1}{\label{#2}}%
  \@esphack
}
\makeatother

\ifpdf
  \usepackage{underscore}         
  \usepackage[T1]{fontenc}        
\else
  \usepackage{breakurl}           
\fi

\title{Unification in Matching Logic --- Revisited}
\author{Ádám Kurucz
\institute{ELTE Eötvös Loránd University\\
Budapest, Hungary}
\email{cphfw1@inf.elte.hu}
\and
Péter Bereczky
\institute{ELTE Eötvös Loránd University\\
Budapest, Hungary}
\email{berpeti@inf.elte.hu}
\and
Dániel Horpácsi
\institute{ELTE Eötvös Loránd University\\
Budapest, Hungary}
\email{daniel-h@elte.hu}
}

\newtheorem{lemma}{Lemma}
\theoremstyle{definition}
\newtheorem{definition}{Definition}
\newtheorem{theorem}{Theorem}
\newtheorem{example}{Example}
\newtheorem{field}{Property}
\newtheorem{corollary}{Corollary}
\DeclareUnicodeCharacter{03BC}{\ensuremath{\mu}}

\newcommand*{\highlight}{\underline}
\newcommand{\EVar}{\textit{EV}}

\newcommand{\app}[2]{#1 \, #2}
\newcommand{\imp}[2]{#1 \to #2}
\newcommand{\expat}[2]{\exists #1.\, #2}

\newcommand{\defined}[1]{\lceil #1 \rceil}
\newcommand{\total}[1]{\lfloor #1 \rfloor}
\newcommand{\subst}[3]{#1[#3/#2]}
\newcommand{\ctxsubst}[2]{#1[#2]}
\newcommand{\defines}{\stackrel{\text{\normalfont\tiny def}}{=}}
\newcommand{\FV}[1]{\textit{FV}(#1)}
\newcommand{\allpat}[2]{\forall #1.\, #2}

\newcommand{\Capp}{C^{\$}}

\newcommand{\hilbert}{\vdash}
\newcommand{\seqindex}{\mathit{S}}
\newcommand{\deduct}{\vdash_\seqindex}
\newcommand{\Constraints}{\mathit{C}}
\newcommand{\TopConstraint}{\top_{\Constraints}}

\newcommand{\sequent}[4]{{#1} \blacktriangleright {#2} \deduct {#3}} 

\newcommand{\provedc}[3]{#1 \hilbert #3} 
\newcommand{\DEFINEDNESS}{\mathsf{DEF}}

\newcommand{\mkspace}{\vspace{0.3cm}}

\newcommand{\singletonUP}[2]{\langle #1,\ #2 \rangle}

\newcommand{\insertUPpair}[3]{#1 \triangleleft \singletonUP{#2}{#3}} 
\newcommand{\toPredUP}[1]{\phi^{#1}}

\newcommand{\rulearrow}[1]{\overset{\textbf{#1}}{\Rightarrow}}
\newcommand{\rulephantom}{\vphantom{\rulearrow{X}}}
\newcolumntype{\ruletablespec}{ >{\(}r<{\)} @{\hspace{1ex}} >{\(}c<{\)} }

\def\titlerunning{Unification in Matching Logic --- Revisited}
\def\authorrunning{Á. Kurucz, P. Bereczky, D. Horpácsi}
\begin{document}
\maketitle

\begin{abstract}
Matching logic is a logical framework for specifying and reasoning about programs using pattern matching semantics. A pattern is made up of a number of structural components and constraints. Structural components are syntactically matched, while constraints need to be satisfied. Having multiple structural patterns poses a practical problem as it requires multiple matching operations. This is easily remedied by unification, for which an algorithm has already been defined and proven correct in a sorted, polyadic variant of matching logic. This paper revisits the subject in the applicative variant of the language while generalising the unification problem and mechanizing a proven-sound solution in Coq.
\end{abstract}

\section{Introduction}

First-order unification is the process of solving equations between first-order terms (symbolic expressions). In particular, unification of first-order terms $t_1$ and $t_2$ yields a substitution $\sigma$ such that it makes the two terms syntactically equal (identical): $t_1 \sigma = t_2 \sigma$. Unification plays a crucial role in automatic theorem proving via first-order resolution as well as in term rewriting.

Matching logic is a general logical framework for specifying and reasoning about programming languages and programs using pattern matching semantics. In matching logic, formulas (called patterns) are evaluated to a subset of a domain, and those evaluating to the full set are valid. Matching logic is the formal meta-language of the K framework~\cite{ChenTrustworthyK}, hosting programming language semantics in terms of constrained rewrite systems. Formal reasoning in these theories requires a sound unification algorithm supporting, amongst others, the reduction of patterns specifying program states.

Despite the fact that matching logic is not an equational logic, syntactic unification has a semantic counterpart in it, based on the notion of equality derived from other, lower-level operations. Given two matching logic patterns $\varphi_1$ and $\varphi_2$, their semantic unifier is a pattern $\varphi$ such that $\varphi \rightarrow \varphi_1 = \varphi_2$, where ``$=$'' denotes semantic equality rather than syntactic equality.

In~\cite{unificationpaper} Arusoaie et al. show that in a sorted, polyadic variant of matching logic~\cite{ROSU2017LMCS}, semantic unification can be derived from syntactic unification: tailoring a well-known rule-based unification algorithm, they obtain the most general unifier substitution $\sigma$, from which they trivially construct a semantic unifier $\phi^\sigma$ and then they show that for patterns $t_1$ and $t_2$, $t_1 \sigma = t_2 \sigma$ implies $\phi^\sigma \leftrightarrow t_1 = t_2$, allowing the equality pattern to be replaced by the unification pattern due to the congruence rule admitted by matching logic.

Following the footsteps of Arusoaie et al. we adapt the syntactic unification algorithm to the applicative, unsorted variant of matching logic and derive the semantic unifier from the solution. Similarly to their work, we assume that the theory specifies a term algebra and the patterns represent constrained first-order terms (i.e., symbols are injective constructor functions). Furthermore, derived from the soundness of the semantic unification, we show that $t_1 \land t_2 = t_1 \land \phi^\sigma$, which allows multiple term patterns to be merged into a constrained term pattern, enabling more effective pattern matching and utilization of SMT solvers. The latter is based on the fact that in matching logic, $t_1 \land t_2$ is equivalent to $t_1 \land t_1 = t_2$, that is, the conjunction of two term patterns can be replaced by one of the term patterns and the predicate stating their semantic equality---such patterns are common in automatic theorem proving in the K framework.

However, compared to related work, the novelty of our approach lies in the fact that we prove the soundness of the unification on a fully syntactical basis, using a sound sequent calculus for matching logic~\cite{proofmodepaper}. Consequently, unlike in the work of Arusoaie et al., we do not need to synthesize proofs for particular unifications as our constructive soundness proof justifies the derivability of the equivalence between the unified patterns and the semantic unifier.

In this work we make the following contributions:

\begin{itemize}
\item The definition of a generalised, abstract unification problem;
\item A rule-based unification algorithm for the unsorted, applicative variant of matching logic;
\item Proof of the soundness of the unification algorithm, using a single-conclusion sequent calculus;
\item Machine-checked implementation of the above results.
\end{itemize}

The rest of the paper is structured as follows. In \Cref{sec:background} we introduce matching logic and the unification problem for first-order terms. Then in \Cref{sec:related} we discuss existing solutions to solving unification and the related work for matching logic in particular. Thereafter, \Cref{sec:unification} defines abstract unification problems, a rule-based solution, and we present the proof of the soundness of the solution. Finally, \Cref{sec:conclusion} summarises the results and concludes.

\section{Background}
\label{sec:background}

In this section, we provide a general introduction to matching logic~\cite{mlexplained} and we also recall some theorems and meta-theorems of matching logic that we will build upon when proving properties of unifier patterns. At the end of this section, we overview the unification problem for first-order terms in general.

\subsection{Matching Logic}

This section introduces an applicative variant of matching logic~\cite{mlexplained,matchingmulogic} which is the language hosting the unification problem explained in this paper.

\subsubsection{Syntax}

Matching logic is a minimal yet expressive language. Its syntax is parametric in the so-called signature, containing constant symbols and variables.

\begin{definition}[Signature]
In a simplified setting, a matching logic signature is a pair $(\EVar,
\Sigma)$, where
\begin{itemize}
\item $\EVar$ is a countably infinite set of element variables (denoted with $x, y, \dots$);
\item $\Sigma$ is a countable set of constant symbols (denoted with $f, g, \dots$).
\end{itemize}
Whenever $\EVar$ is understood from the context, only $\Sigma$ is used to denote the whole signature.
\end{definition}

The only syntactic category of the language is \emph{patterns}, parametrised by a particular signature.

\begin{definition}[Pattern]
Given a signature $(\EVar,
 \Sigma)$, the following rules define patterns:
\begin{equation*}
\varphi ::= x 
\mid f \mid \app{\varphi_1}{\varphi_2} \mid \bot \mid \imp{\varphi_1}{\varphi_2} \mid \expat{x}{\varphi}
\end{equation*}
These constructs are in order: (element) variables, symbols, application of patterns, falsum (bottom) and the traditional first-order operations (implication and existential quantification). Application is left associative and binds the tightest. Implication is right associative. The scope of the $\exists$ binder extends as far to the right as possible. We use $\FV{\varphi}$ to refer to the set of free variables of some pattern $\varphi$.
\end{definition}

Note that the full version of matching logic also contains least-fixpoint ($\mu$) patterns~\cite{mlexplained}, but they are not relevant for this work, thus for the sake of readability we omit them from the current presentation. The reason for this is that most of the patterns presented here are term patterns, which do not contain fixpoints syntactically. We note that the formalisation~\cite{unifpr} includes $\mu$-patterns too, but some mechanised theorems and rules (e.g., \ref{rule:rewritelocal}) are restricted to work with $\mu$-free patterns currently.

Matching logic is intentionally minimal, but can derive several other, well-known constructs as syntactic sugar. For instance, the following lines define negation, disjunction, conjunction, top, equivalence, and universal quantification.

\begin{align*}
	\lnot \varphi &\defines \imp{\varphi}{\bot} & \varphi_1 \lor \varphi_2 &\defines \lnot \imp{\varphi_1}{\varphi_2} \\
	\varphi_1 \land \varphi_2 &\defines \lnot (\lnot \varphi_1 \lor \lnot \varphi_2) & \top &\defines \lnot \bot \\
	\varphi_1 \leftrightarrow \varphi_2 &\defines (\varphi_1 \rightarrow \varphi_2) \land (\varphi_2 \rightarrow \varphi_1) & \allpat{x}{\varphi} &\defines \lnot \expat{x}{\lnot \varphi}
\end{align*}

We define substitutions on patterns in the usual way.

\begin{definition}[Substitution]\label{def:normalsubs}
We denote substitutions as $\subst{\varphi}{x}{\psi}$ to say that we replace every free occurrence of $x$ in $\varphi$ with $\psi$.
\end{definition}

Special patterns are pattern contexts and application contexts, which are always used in substitutions.

\begin{definition}[Pattern context, application context]
A pattern context, denoted with $C$ is a pattern with a distinguished variable $\square$, called a hole. When substituting	$\varphi$ into a hole, we simply write $\ctxsubst{C}{\varphi}$ denoting $\subst{C}{\square}{\varphi}$. When a context contains only applications from the root to the (only) $\square$, we call it an application context. Application contexts are defined with the following grammar:
\begin{equation*}
\Capp := \square \mid \app{\varphi}{\Capp} \mid \app{\Capp}{\varphi}
\end{equation*}
\end{definition}

Contexts are useful for defining theorems that are only concerned with a smaller part of a pattern while the rest is irrelevant. Contexts are used for example to substitute equal patterns inside a bigger context (see rule \ref{rule:rewritelocal}). This definition of contexts is slightly different from the definition used in~\cite{mlexplained,matchingmulogic}, but the same proofs can be expressed with both variants\footnote{This correspondence is proven in the Coq formalisation~\cite[\texttt{matching-logic/src/ProofMode/Misc.v}]{unifpr}.}.

\subsubsection{Semantics}

Matching logic semantics is based on pattern matching. Patterns are
interpreted as a set of domain elements that match them.

\begin{definition}[Model]
A model is a tuple $(M, \_{\cdot}\_, M_f)$, where
\begin{itemize}
\item $M$ is the non-empty carrier set (domain), which contains all elements a pattern may evaluate to,
\item $\_{\cdot}\_ : M \rightarrow M \rightarrow \mathcal P(M)$ is a binary function that serves as the interpretation of application,
\item $M_f \subseteq M$ is the (indexed) interpretation of symbols $f \in \Sigma$.
\end{itemize}
\end{definition}

We extend the application function to subsets of the domain in a pointwise manner:
\begin{equation*}
	A \cdot B = \bigcup_{\mathclap{\substack{a \in A \\ b \in B}} } a \cdot b
\end{equation*}

Let us define the semantics of matching logic formulas informally. Given a signature and a model, we can define interpretations for
patterns. Element variables are interpreted as singleton sets, while the meaning of symbols is given by $M_f$. In the case of application, the meanings of the subpatterns are combined using the $\_{\cdot}\_$ function. $\bot$ is matched by the empty set and the implication $\imp{\varphi_1}{\varphi_2}$ is matched by elements that match $\varphi_2$ if they match $\varphi_1$. $\expat{x}{\varphi}$ is matched by all instances of $\varphi$ as $x$ ranges over $M$. This means, for example, that $\expat{x}{x}$ is matched by $M$. The semantics of the derived operations follow from these.

\begin{definition}[Semantic consequence]
A pattern is validated by a model if the pattern evaluates to the entire carrier set of the model. A model validates a set of patterns, if it validates all patterns in the set. In particular, we use $\Gamma \vDash \varphi$ to denote that the patterns in $\Gamma$ (i.e., the set of axioms) validate a pattern $\varphi$ whenever all models validating $\Gamma$ also validate $\varphi$.
\end{definition}

\begin{definition}[Term patterns and predicate patterns]
We distinguish two classes of patterns that will be important in the later sections of this paper. Specifically, \emph{term patterns} evaluate to singleton sets (they behave like terms in first-order logic), and \emph{predicate patterns} evaluate either to the empty set or to the full carrier set (they mimic formulas in first-order logic).
\end{definition}

\subsubsection{Sequent Calculus for Matching Logic}

Matching logic has a sound Hilbert-style proof system~\cite{mlexplained} with a handful of simple
rules, and a sequent calculus~\cite{proofmodepaper} which is sound and complete w.r.t. the proof system. In this section, we briefly summarize this calculus, and we construct our proofs with it in later sections. We denote provability in the Hilbert-system with $\provedc{\Gamma}{c}{\varphi}$, and next we define sequents.

\begin{definition}[Sequents]
  \label{def:sequent}
	A \emph{sequent} is a triple $\sequent{\Gamma}{\Delta}{\psi}{c} \,$ \footnote{Note that $\deduct$ denotes a derivation in the sequent calculus, whereas $\hilbert$ denotes a derivation in the Hilbert-style proof system.}, where
  \begin{itemize}
  \item $\Gamma$ is a (possibly infinite) set of patterns, called a \emph{theory};
  \item $\Delta$ is a finite (comma-separated) list of patterns, called \emph{antecedent} or \emph{local context};
  \item $\psi$ is a pattern, called \emph{succedent} or \emph{conclusion}.
  \end{itemize}
\end{definition}

\begin{figure}[t]
    \begin{subfigure}{\textwidth}
        \centering
        \begin{prooftree}
            \hypo{\provedc{\Gamma}{c}{\psi}}
            \infer1[\labeltext{\textsc{Inherit}}{rule:inherit}]{\sequent{\Gamma}{[]}{\psi}{c}}
        \end{prooftree}
        \hfill
        \begin{prooftree}
            \hypo{\sequent{\Gamma}{\Delta_1, \Delta_2}{\psi}{c}}
            \infer1[\labeltext{\textsc{Weaken}}{rule:clear}]{\sequent{\Gamma}{\Delta_1, \varphi, \Delta_2}{\psi}{c}}
        \end{prooftree}
        \hfill
        \begin{prooftree}
            \hypo{\sequent{\Gamma}{\Delta_1}{\varphi}{c}}
            \hypo{\sequent{\Gamma}{\Delta_1,\varphi,\Delta_2}{\psi}{c}}
            \infer2[\labeltext{\textsc{Cut}}{rule:assert}]{\sequent{\Gamma}{\Delta_1,\Delta_2}{\psi}{c}}
        \end{prooftree}
        \caption{Technical and structural inference rules}
        \vspace*{0.5cm}
    \end{subfigure}
    \begin{subfigure}{\textwidth}
        \centering
        \begin{prooftree}
            \infer0[\labeltext{\textsc{Hyp}}{rule:exact}]{\sequent{\Gamma}{\Delta_1,\varphi,\Delta_2}{\varphi}{c}}
        \end{prooftree}
        \hfill
        \begin{prooftree}
            \infer0[\labeltext{$(\vdash\bot)$}{rule:botelim}]{\sequent{\Gamma}{\Delta_1,\bot,\Delta_2}{\psi}{c}}
        \end{prooftree}
        
        \mkspace
        
        \begin{prooftree}
            \hypo{\sequent{\Gamma}{\Delta_1, \Delta_2}{\varphi_1}{c}}
            \hypo{\sequent{\Gamma}{\Delta_1, \varphi_2, \Delta_2}{\psi}{c}}
            \infer2[\labeltext{$(\to\vdash)$}{rule:elimimpl}]{\sequent{\Gamma}{\Delta_1, \varphi_1 \to \varphi_2, \Delta_2}{\psi}{c}}
        \end{prooftree}
        \hfill
        \begin{prooftree}
            \hypo{\sequent{\Gamma}{\Delta, \varphi}{\psi}{c}}
            \infer1[\labeltext{$(\vdash\to)$}{rule:introimpl}]{\sequent{\Gamma}{\Delta}{\varphi \to \psi}{c}}
        \end{prooftree}
        
        \mkspace
        
        \begin{prooftree}
            \hypo{\sequent{\Gamma}{\Delta_1, \varphi_1, \varphi_2, \Delta_2}{\psi}{c}}
            \infer1[\labeltext{$(\land\vdash)$}{rule:destructand}]{\sequent{\Gamma}{\Delta_1, \varphi_1 \land \varphi_2, \Delta_2}{\psi}{c}}
        \end{prooftree}
        \hfill
        \begin{prooftree}
            \hypo{\sequent{\Gamma}{\Delta}{\psi_1}{c}}
            \hypo{\sequent{\Gamma}{\Delta}{\psi_2}{c}}
            \infer2[\labeltext{$(\vdash\land)$}{rule:splitand}]{\sequent{\Gamma}{\Delta}{\psi_1 \land \psi_2}{c}}
        \end{prooftree}
        
        \mkspace
        
        \begin{minipage}{.65\linewidth}
            \begin{prooftree}
                \hypo{\sequent{\Gamma}{\Delta_1, \varphi_1, \Delta_2}{\psi}{c}}
                \hypo{\sequent{\Gamma}{\Delta_1, \varphi_2, \Delta_2}{\psi}{c}}
                \infer2[\labeltext{$(\lor\vdash)$}{rule:destructor}]{\sequent{\Gamma}{\Delta_1, \varphi_1 \lor \varphi_2, \Delta_2}{\psi}{c}}
            \end{prooftree}
        \end{minipage}
        \begin{minipage}{.31\linewidth}
            \begin{prooftree}
                \hypo{\sequent{\Gamma}{\Delta}{\psi_1}{c}}
                \infer1[\labeltext{$(\vdash\lor_{L})$}{rule:leftor}]{\sequent{\Gamma}{\Delta}{\psi_1 \lor \psi_2}{c}}
            \end{prooftree}
            
            \mkspace
            
            \begin{prooftree}
                \hypo{\sequent{\Gamma}{\Delta}{\psi_2}{c}}
                \infer1[\labeltext{$(\vdash\lor_{R})$}{rule:rightor}]{\sequent{\Gamma}{\Delta}{\psi_1 \lor \psi_2}{c}}
            \end{prooftree}
        \end{minipage}
        \caption{Inference rules for propositional reasoning}
        \vspace*{0.5cm}
    \end{subfigure}
    
    \begin{subfigure}{\textwidth}
        \centering
        
        \begin{prooftree}
            \hypo{\sequent{\Gamma}{\Delta_1, \subst{\varphi}{x}{y}, \Delta_2}{\psi}{c}}
            \infer1[\labeltext{$(\forall\vdash)$}{rule:elimall}]{\sequent{\Gamma}{\Delta_1, \forall x.\, \varphi, \Delta_2}{\psi}{c}}
        \end{prooftree}
        \hfill
        \begin{prooftree}
            \hypo{\sequent{\Gamma}{\Delta}{\subst{\psi}{x}{y}}{c}}
            \infer1[$y \notin \FV{\Delta, \forall x.\, \psi}\quad$\labeltext{$(\vdash\forall)$}{rule:introall}]{\sequent{\Gamma}{\Delta}{\forall x .\, \psi}{c}}
        \end{prooftree}
        
        \mkspace
        
        \begin{prooftree}
            \hypo{\sequent{\Gamma}{\Delta_1, \subst{\varphi}{x}{y}, \Delta_2}{\psi}{c}}
            \infer1[$y \notin \FV{\Delta, \exists x.\, \varphi, \psi}\quad$\labeltext{$(\exists\vdash)$}{rule:elimex}]{\sequent{\Gamma}{\Delta_1, \exists x.\, \varphi, \Delta_2}{\psi}{c}}
        \end{prooftree}
        \hfill
        \begin{prooftree}
            \hypo{\sequent{\Gamma}{\Delta}{\subst{\psi}{x}{y}}{c}}
            \infer1[\labeltext{$(\vdash\exists)$}{rule:introex}]{\sequent{\Gamma}{\Delta}{\exists x.\, \psi}{c}}
        \end{prooftree}
        \caption{Inference rules for first-order reasoning}
    \end{subfigure}
    
    \caption{Sequent calculus for matching logic (theory-independent rules)}
    \label{fig:PropSequentRules}
\end{figure}

Next, we recall the sequent calculus in~\Cref{fig:PropSequentRules} and refer to~\cite{proofmodepaper} for further details. Essentially, these rules mimic a standard single-conclusion sequent calculus for first-order logic, except that they use a concept of \emph{local context} and do not affect the theory ($\Gamma$) directly. Informally, the meaning of the sequent $\sequent{\Gamma}{\varphi_1, \dots, \varphi_n}{\psi}{c}$ can be expressed as an expansion to $\provedc{\Gamma}{c}{\varphi_1 \to \dots \to \varphi_n \to \psi}$, which allows the premises to be separated from the conclusion without using a deduction theorem. With this notion in mind, we can see that the rule \ref{rule:introimpl} is not a deduction theorem of matching logic (we refer to \Cref{sec:definedness} for further details), but in practice it functions as a rule turning implication premises to hypotheses. The sequent calculus accommodates another rule for the deduction meta-theorem (see~\ref{rule:dt}), but in the general case (when the fixed-point operator is present in the language) that comes with technical side conditions~\cite{proofmodepaper}, making it more difficult to apply.

We also highlight the rule~\ref{rule:inherit}, which allows us to lift any proofs done with the Hilbert-style proof system ($\hilbert$) into the sequent calculus. In later sections, we use this rule implicitly to lift Hilbert-style proofs of theorems while presenting a sequent calculus proof (note that with the repeated use of \ref{rule:clear} we can unify the local context of the lifted theorem and the current proof state). We also note that it is also possible to express all sequent calculus proof in the Hilbert-style proof system too (\Cref{lem:correspondenceLemmaBasic}). For further insights on the correspondence between the calculus and the proof system, we refer to the paper on the calculus~\cite{proofmodepaper} and its implementation~\cite[\texttt{matching-logic/src/ProofMode}]{unifpr}.

Next, we encode the theory of definedness (also featuring equality) in matching logic, and extend the sequent calculus with rules using definedness and equality, notably deduction and rewriting.

\subsubsection{Definedness}\label{sec:definedness}

We recall the formal specification for the theory of \emph{definedness} from~\cite{mlexplained} in \Cref{fig:defined}. The section ``Symbol'' includes the new symbols added to the signature, while ``Notation'' describes the derived notations used in the theory. Finally, ``Axiom'' describes the axiom schemas that specify the symbols of the theory. In the rest of the paper, we will refer to the axiom set of the definedness theory with $\Gamma^\DEFINEDNESS$. \\

\begin{figure}
\noindent\fbox{\begin{minipage}{.98\textwidth}\small
	\textbf{spec} DEFINEDNESS\\
	\-\ \ Symbol: $\defined{\_}$\\
	\-\ \ Notation:

	\vspace*{-15mm}

	\begin{align*}
		\defined{\varphi} &\defines \app{\defined{\_}}{\varphi} & \total{\varphi} &\defines \lnot\defined{\lnot\varphi} \\
		\varphi_1 = \varphi_2 &\defines \total{\varphi_1 \leftrightarrow \varphi_2} & \varphi_1 \ne \varphi_2 &\defines \lnot (\varphi_1 = \varphi_2) \\
		\varphi_1 \in \varphi_2 &\defines \defined{\varphi_1 \land \varphi_2} & \varphi_1 \notin \varphi_2 &\defines \lnot (\varphi_1 \in \varphi_2) \\
		\varphi_1 \subseteq \varphi_2 &\defines \total{\varphi_1 \rightarrow \varphi_2} & \varphi_1 \nsubseteq \varphi_2 &\defines \lnot (\varphi_1 \subseteq \varphi_2)
	\end{align*}

	\vspace*{-5mm}

	\-\ \ Axiom:\\
	\-\ \ \ \ (\textsc{Definedness})\ \ \ $\defined{x}$\\
	\textbf{endspec}
\end{minipage}}
\caption{Specification of definedness}
\label{fig:defined}
\end{figure}

This theory introduces a symbol $\defined{\_}$, a number of notations on top of this symbol, and one axiom which embodies the meaning of the definedness symbol. The axiom states that all (element) variables are defined, i.e., the definedness applied to a variable is always satisfied. Intuitively, $\defined{\varphi}$ is satisfied if $\varphi$ matches at least one element. Based on definedness, we can derive \emph{totality} (denoted as $\total{\varphi}$), and the usual notion of equality, membership and set inclusion. Intuitively, $\total{\varphi}$ is satisfied if $\varphi$ matches all elements of the domain.
With equality, we can syntactically express when a pattern is a term pattern (also called functional pattern), or a predicate pattern. We can also show when symbol $f$ behaves like an $n$-ary function, or as an $n$-ary predicate:

\begin{align*}
    \expat{x}{\varphi = x}\label{ax:functional}\tag{Functional Pattern}\\
    \varphi = \bot \lor \varphi = \top \label{ax:predicate}\tag{Predicate Pattern}\\
    \allpat{x_1}{\dots \allpat{x_n}{\expat{y}{\app{\app{f}{x_1}}{ \cdots x_n} = y}}}    \label{ax:fun}\tag{Function}\\
    \allpat{x_1}{\dots \allpat{x_n}{\app{\app{g}{x_1}}{ \cdots x_n} = \top \lor\app{\app{g}{x_1}}{ \cdots x_n} = \bot }}    \label{ax:pred}\tag{Predicate}	
\end{align*}

For reasoning about equality, we use the rules presented in \Cref{fig:DefSequentRules}. All rules of the sequent calculus are proven sound in~\cite{proofmodepaper}. Note that both the deduction rule and the rewriting rule (replacing equal patterns) are based on the deduction theorem and the equality elimination theorem, which in their general form come with complex technical side conditions intentionally neglected in this presentation; for details about the technical side conditions, we refer to~\cite{matchingmulogic}.

\renewcommand*{\thefootnote}{\fnsymbol{footnote}}

\begin{figure}[htb]
  \begin{prooftree}
      \hypo{\sequent{\Gamma \cup \{ \varphi \}}{\Delta}{\psi}{c}\ \footnotemark}
      \infer1[\labeltext{\textsc{Deduction}}{rule:dt}]{\sequent{\Gamma}{\total{\varphi}, \Delta}{\psi}{\TopConstraint}}
  \end{prooftree}
  \hspace{1em}
  \begin{prooftree}
    \infer0[\labeltext{$(\vdash=)$}{rule:reflexivity}]{\sequent{\Gamma}{\Delta}{\psi = \psi}{\TopConstraint}}
  \end{prooftree}
  \hspace{1em}
  \begin{prooftree}
      \hypo{\sequent{\Gamma}{\Delta_1, \varphi_1 = \varphi_2, \Delta_2}{\ctxsubst{C}{\varphi_2}}{\TopConstraint}}
      \infer1[\labeltext{$(=\vdash)$}{rule:rewritelocal}]{\sequent{\Gamma}{\Delta_1, \varphi_1 = \varphi_2, \Delta_2}{\ctxsubst{C}{\varphi_1}}{\TopConstraint}}
  \end{prooftree}    
  \caption{Deduction; and rules about equality. These rules assume that $\Gamma^\DEFINEDNESS \subseteq \Gamma$.}
  \label{fig:DefSequentRules}
\end{figure}

\footnotetext{The proof does not use existential generalization on free variables of $\psi$.}

\renewcommand*{\thefootnote}{\arabic{footnote}}

Next, we state important theorems and meta-theorems of matching logic which we will rely on when reasoning about the correctness of unification (we refer to~\cite{bereczky2022mechanizing,mlexplained} for detailed explanations and proofs).

\subsubsection{Essential Theorems}

First of all, we state the correspondence theorem that allows us to use Hilbert-style proofs and sequent calculus proofs interchangeably. In particular, this allows us to reason about Hilbert-style provability by using the sequent calculus.

\begin{theorem}[Correspondence]\label{lem:correspondenceLemmaBasic}
    For all theories $\Gamma$, and patterns $\varphi_1, \dots, \varphi_k, \psi$, the derivability statement
    $\sequent{\Gamma}{\varphi_1,\ldots,\varphi_k}{\psi}{c}$ holds if and only if
    $\provedc{\Gamma}{c}{\varphi_1 \to \dots \to \varphi_k \to \psi}$.
\end{theorem}

Now we state fundamental (meta-)theorems heavily used in the subsequent sections.

\begin{theorem}[Congruence]
    For all patterns $\varphi_1, \varphi_2$, pattern contexts $C$ and theory $\Gamma$, if $\Gamma \vdash \varphi_1 \leftrightarrow \varphi_2$ then $\Gamma \vdash \ctxsubst{C}{\varphi_1} \leftrightarrow \ctxsubst{C}{\varphi_2}$.
\end{theorem}

\begin{lemma} \label{phiimpldefphi}
	Let $\Gamma$ be a theory and $\varphi$ a pattern. Then $\Gamma \vdash
	\varphi \rightarrow \defined \varphi$.
\end{lemma}

Note that~\Cref{phiimpldefphi} is a consequence of the \textsc{Definedness} axiom and existential generalization.

\begin{lemma} \label{membimplequal}
	Let $\Gamma$ be a theory and $\varphi_1$ and $\varphi_2$ functional
	patterns. Then $\Gamma \vdash \varphi_1 \in \varphi_2 \rightarrow
	\varphi_1 = \varphi_2$.
\end{lemma}

The following theorem is used to extract and insert an extra condition into both sides of an equivalence; note that it also holds for $\varphi_1$ on the right side, by the commutativity of conjunction.

\begin{lemma} \label{ecfequiv}
	Let $\Gamma$ be a theory and $\varphi_1$, $\varphi_2$, and $\varphi_3$
	be patterns. Then $\Gamma \vdash (\varphi_1 \land \varphi_2
	\leftrightarrow \varphi_1 \land \varphi_3) \leftrightarrow (\varphi_1
	\rightarrow \varphi_2 \leftrightarrow \varphi_3)$.
\end{lemma}

\Cref{ecfequiv} can be extended to work with equality using deduction:

\begin{lemma} \label{ecfequal}
    For all theories $\Gamma$, predicate patterns $\varphi_1$, and patterns
    $\varphi_2, \varphi_3$, $\Gamma \vdash
    (\varphi_1 \land \varphi_2 = \varphi_1 \land \varphi_3) \leftrightarrow
    (\varphi_1 \rightarrow \varphi_2 = \varphi_3)$ holds.
\end{lemma}

\subsection{Unification}
\label{sec:unificationbckgr}

Unification~\cite{dorel6thref} is the process of solving equations between symbolic expressions. The solution of the unification problem is a substitution that maps variables to expressions, which, when applied to the symbolic expressions, makes them syntactically equal (identical). In this subsection, we recall some essential concepts (mainly from~\cite{dorel6thref}) which we will use when discussing unification in matching logic.

The following definitions assume that we work using terms ($t$) in a term algebra constructed with function symbols ($f$) and variables ($x$).

\begin{definition}[Substitution in unification]
A substitution is a list of bindings $\sigma = \{x_1 \mapsto t_1,\ \dots,\ x_n \mapsto t_n\}$, mapping variables to terms. This is similar to the one in \Cref{def:normalsubs}, however, unlike that one, this version is not limited to a single variable and pattern.
\end{definition}

The application of a substitution of a term ($t \sigma$) is defined as usual, in a recursive descent manner on applications.

\begin{definition}[Composition of substitutions]
	Let $\sigma = \{x_1 \mapsto t_1,\ \dots,\ x_n \mapsto t_n\}$ and $\eta$
	be two substitutions. The composition of these, written as $\sigma\eta$
	can be defined as $\sigma\eta \defines \{x_1 \mapsto t_1\eta,\ \dots,\
	x_n \mapsto t_n\eta\}$, i.e. we apply the substitution $\eta$ to every
	pattern on the right side of $\sigma$.
\end{definition}

Next we define how to determine if two substitutions are equal.

\begin{definition}[Equality of substitutions]
	Two substitutions $\sigma$ and $\eta$ are equal if they are
	extensionally equal: $t\sigma = t\eta$, for all $t$, that is if their
	effect is the same when applied to any term.
\end{definition}

We can combine these two operations to define when one substitution is more
general than another.

\begin{definition}[More general substitution]
	We say that $\sigma$ is more general than $\eta$, written as $\sigma
	\le \eta$ if there is a substitution $\theta$ such that $\sigma\theta =
	\eta$.
\end{definition}

The unification algorithm will produce a substitution that, when applied to the input terms, produces the same term. We call this substitution a unifier.

\begin{definition}[Unifier, unifiable terms]
	A substitution $\sigma$ is called a unifier of two patterns $t_1$ and
	$t_2$ if $t_1\sigma = t_2\sigma$. If there exists a unifier of two terms, we call the terms unifiable.
\end{definition}

However, there may be infinitely many unifiers that only differ in some
insignificant details. For example, if the terms are $f \, x$ and $f \, y$,
we could replace $x$ and $y$ with any term, but all that matters is that
they need to be the same. To solve this, we introduce the concept of the
most general unifier, using the more general relation from above.

\begin{definition}[Most general unifier]
	A unifier $\sigma$ is called the most general if for any unifier
	$\theta$, $\sigma$ is more general than $\theta$ ($\sigma \le \theta$).
\end{definition}

It follows from the definition of the more general unifier that any other
solution may be obtained from this by instantiation. In this paper we will
not consider higher order unification (i.e. variables may not stand in for
functions).

\begin{definition}[Predicate of a substitution]
	It is possible to extract a predicate from a substitution $\sigma$, denoted $\phi^\sigma$. If $\sigma = \{x_1 \mapsto t_1,\ \dots,\ x_n \mapsto t_n\}$, then $\phi^\sigma \defines x_1 = t_1 \land \dots \land x_n = t_n$.
\end{definition}

In the following sections, we will solve semantic unification by solving syntactic unification, producing a most general unifier and turning it into a semantically unifying pattern in matching logic.

\section{Related Work}
\label{sec:related}

\subsection{Rule-based Unification}

There are several algorithms for solving unification problems. We will focus on the rule-based approach that mimics the recursive descent algorithm computing the most general unifier and review the fundamental work presented in~\cite{dorel6thref}. In particular, \cite{dorel6thref} introduces the so-called \emph{unification problem}, which is either a set of pairs of term patterns, or a special symbol $\bot$, representing an unsolvable problem\footnote{For brevity, from this point on we use $\bot$ do denote the failed unification problem instead the falsum pattern.}.

The algorithm is represented by a \emph{unification step} relation, denoted as
$P \Rightarrow P'$ for some $P$ and $P'$ unification problems. This
relation is inductively defined by the rules of the algorithm. See
\Cref{fig:polyrules} for the rules. The reflexive-transitive closure of
this relation is written as $P \Rightarrow^* P'$. It is assumed that the
left side of $\Rightarrow$, $\Rightarrow^*$, and $\cup$ is not $\bot$.

\begin{table}[htb]
	\centering
	\begin{tabular}{ >{\bfseries}l >{\(}l<{\)} }
		Delete & P \cup \{(t,\ t)\} \Rightarrow P \\
		Decomposition & P \cup \{(f(t_1,\ \dots,\ t_n),\ f(u_1,\ \dots,\ u_n))\} \Rightarrow P \cup \{(t_1,\ u_1),\ \dots,\ (t_n,\ u_n)\} \\
		Symbol clash & P \cup \{(f(t_1,\ \dots,\ t_n),\ g(u_1,\ \dots,\ u_m))\} \Rightarrow \bot \\
		Orient & P \cup \{(t,\ x)\} \Rightarrow P \cup \{(x,\ t)\} \quad \text{if } t \notin \EVar \\
		Occurs check & P \cup \{(x,\ t)\} \Rightarrow \bot \quad \text{if } x\in FV(t) \\
		Elimination & P \cup \{(x,\ t)\} \Rightarrow P\{x \mapsto t\} \cup \{(x,\ t)\} \quad \text{if } x \notin FV(t) \\
	\end{tabular}
	\caption{Rules of the unification algorithm}
	\label{fig:polyrules}
\end{table}

\begin{definition}[Solved form]
A unification problem is considered to be in solved form if it is $\bot$	or a set $\{(x_1,\ t_1),\ \dots,\ (x_n,\ t_n)\}$, where $x_i \notin t_j$, for any $1 \le i,j \le n$ (i.e. the first term in the pairs is a variable that does not appear in the second term of any of the pairs).
\end{definition}

A unification problem that is in solved form and is not $\bot$ can be seen
as a substitution. If we start with a unification problem with just the two
terms we want to unify as a pair, then apply the rules of the algorithm, we
will either end up with $\bot$, in which case the terms are not unifiable,
or a unification problem whose corresponding substitution is the most
general unifier of the terms. This expresses the \emph{soundness} of the algorithm. Note that the opposite of the above statement is also true:

\begin{theorem} \label{thm:convenient}
	If $\sigma$ is the most general unifier of $t_1$ and $t_2$, then there exists a unification problem $P$ such that
	$\{(t_1, t_2)\} \Rightarrow^* P$ and $P$ is in solved form,
	with its corresponding substitution being $\sigma$.
\end{theorem}

Let us demonstrate rule-based unification with an example, which we will revisit in later sections for the sake of comparison.

\begin{example} \label{polyex}
	The following is an exhaustive application of the unification rules on the unification problem constructed with terms $f(x,\ 
	g(1),\ g(z))$ and $f(g(y),\ g(y),\ g(g(x)))$.

	\centering
	\begin{tabular}{\ruletablespec}
		\{\highlight{(f(x,\ g(1),\ g(z)),\ f(g(y),\ g(y),\ g(g(x))))}\} & \rulearrow{Decomposition} \\
		\{(x,\ g(y)),\ \highlight{(g(1),\ g(y))},\ (g(z),\ g(g(x)))\}   & \rulearrow{Decomposition} \\
		\{(x,\ g(y)),\ (1,\ y),\ \highlight{(g(z),\ g(g(x)))}\}         & \rulearrow{Decomposition} \\
		\{(x,\ g(y)),\ \highlight{(1,\ y)},\ (z,\ g(x))\}               & \rulearrow{Orient} \\
		\{(x,\ g(y)),\ \highlight{(y,\ 1)},\ (z,\ g(x))\}               & \rulearrow{Elimination} \\
		\{\highlight{(x,\ g(1))},\ (y,\ 1),\ (z,\ g(x))\}               & \rulearrow{Elimination} \\
		\{(x,\ g(1)),\ (y,\ 1),\ (z,\ g(g(1)))\}                        & \rulephantom \\
	\end{tabular}
\end{example}

\subsection{Unification in Matching Logic}

In matching logic, symbolic expressions are encoded as term patterns that are made with applications of constructor function symbols. Unification has already been investigated in matching logic; in particular, Arusoaie et al.~\cite{unificationpaper} adapted the results of~\citep{dorel6thref} to the sorted, polyadic variant of matching logic (where applications are of form $f(x_1,\ \dots,\ x_n)$). In this paper, we do a similar adaptation, but to the unsorted, applicative variant of the logic, where application is a binary operation: $\varphi_1 \, \varphi_2$.

It is to be noted that~\citep{unificationpaper} argues about the soundness of unification on a semantic basis (in terms of the evaluations of patterns), while our work shows a syntactic proof of soundness, based on the sequent calculus for the logic, particularly exploiting the rules for deduction and rewriting.

Although the logic language and therefore the unification problems are somewhat different in our approach, we tried to reuse as much as possible from that of~\citep{unificationpaper}. Specifically, even though we generalise the original unification problem into an abstract data type, we closely follow the approach of turning unification problems and substitutions into predicate patterns in matching logic.

\section{Unification in an Applicative Matching Logic}
\label{sec:unification}

As discussed in~\Cref{sec:related}, the rule-based unification algorithm relies on a data structure called a unification problem. This is a set of pairs in the original presentation~\citep{dorel6thref}, but we generalise it to an abstract data type that is not tied to a specific underlying container. 

\subsection{Abstract Unification Problem}

First of all, we outline the operations that can be abstracted into the constructors of the new type (and were inherited from set theory and substitution theory in the original definition). The constructors of abstract unification problems are as follows ($P$ is a unification problem, $t$, $t_1$, $t_2$ are term patterns, $x$ is a variable):

\begin{itemize}
	\item $\singletonUP{t_1}{t_2}$ creates a new unification problem from a single pair of terms (this method is used to create the initial problem from the input terms that the algorithm is applied to);
	\item $\bot$ represents a unification problem that ``failed'' because the initial terms were not unifiable;
	\item $\insertUPpair{P}{t_1}{t_2}$ is an insert operation that adds a single pair to the data structure;
	\item $\subst P x t$ transforms the data structure by substituting every occurrence of $x$ in every pair by $t$.
\end{itemize}

We chose to define a constructor to create a singleton unification problem, but note that it would have also been equally correct to allow the creation of an empty problem $\emptyset$, in which case singleton could have simply been defined as $\singletonUP{t_1}{t_2} \defines \insertUPpair{\emptyset}{t_1}{t_2}$. Although we rarely need to work with an empty unification problem, it may be created by the algorithm if the input terms are already the same, therefore no substitution is needed to unify them.

There is only a single destructor for the data structure, $\toPredUP P$, which generates a pattern out of all the stored pairs. Intuitively, when the underlying representation uses sets and pairs, the implementation of decomposition would be defined by the following line:

\begin{equation*}
	\toPredUP P \defines \bigwedge_{\mathclap{(t_1,\ t_2) \in P}} t_1 = t_2
\end{equation*}

However, abstract unification problems can be instantiated to various concrete representations. In the following, we provide the specification of the unification problem by determining the behaviour of the constructor and destructor operations, independent of the representation.

\paragraph{Injectivity.}

From now on, we suppose that $\Gamma$ is a theory that contains the theory of definedness and includes axioms constraining all term constructor symbols to be modelled by injective functions; namely, we assume instances of the following axiom scheme for all constructor symbols $f$:

\begin{equation*}
	\forall x_1,\ \dots,\ x_n,\ y_1,\ \dots,\ y_n. f(x_1,\ \dots,\ x_n) = f(y_1,\ \dots,\ y_n) \rightarrow x_1 = y_1 \land \dots \land x_n = y_n
\end{equation*}

Now, the following four specification axioms define the behaviour of the operations of abstract unification problems:

\begin{itemize}
\item The first property states that a singleton problem is equivalent to a single equality made from its constituent pair.
\begin{field}\label{toPredicateSingletonUP}
	$\Gamma \vdash \toPredUP{\singletonUP{t_1}{t_2}} \leftrightarrow t_1 = t_2$
\end{field}
\item The second property says that if a pair is inserted into the problem, its corresponding equality should appear in the predicate as well.
\begin{field}\label{toPredicateInsertUP}
	$\Gamma \vdash \toPredUP{\insertUPpair{P}{t_1}{t_2}} \leftrightarrow t_1 = t_2 \land \toPredUP P$
\end{field}
\item The third property states that substitution propagates through the predicate creation, that is substituting the terms in the data structure results in substituting them in the predicate as well.
\begin{field}\label{toPredicateSubstituteAllUP}
	$\Gamma \vdash \toPredUP{\subst P x t} \leftrightarrow \subst{(\toPredUP P)}{x}{t}$
\end{field}
\item Finally, the fourth property expresses that if we extend a non-$\bot$ unification problem, the result will not be $\bot$.
\begin{field}\label{insertNotBottomUP}
	$\Gamma \vdash P \ne \bot \rightarrow \insertUPpair{P}{t_1}{t_2} \ne \bot$
\end{field}
\end{itemize}

\subsection{Curried Rules for Abstract Unification Problems}

While the related work~\citep{unificationpaper} uses the polyadic version of the logic to define the rules of the unification algorithm, our goal is to create a formalism in the applicative version. Since the latter only has binary function application, every rule that uses this connective in the former must be rephrased. This is done by replacing functions with their curried versions, essentially taking the argument list and applying its elements one at a time:

\begin{equation*}
	f(t_1,\ \dots,\ t_n) \longrightarrow (\app{(\app{(\app{f} {t_1})} {\dots})} {t_n})
\end{equation*}

The rules impacted by this change are \textbf{Decomposition} and \textbf{Symbol clash}. The new ruleset is described in \Cref{fig:apprules}. Notice how, because we always insert one pair in most rules, and two pairs in \textbf{Decomposition}, the previously used union operations can be easily replaced with simple insertions.

\begin{table}[htb]
	\centering
	\begin{tabular}{ >{\bfseries}l >{\(}l<{\)} }
		Delete & \insertUPpair P t t \Rightarrow P \\
		Decomposition & \insertUPpair{P}{\app{t_1} {t_2}}{\app{u_1} {u_2}} \Rightarrow \insertUPpair{\insertUPpair{P}{t_1}{u_1}}{t_2}{u_2} \\
		Symbol clash L & \insertUPpair P f t \Rightarrow \bot \quad \text{if } t \ne f \land t \notin \EVar \\
		Symbol clash R & \insertUPpair P t f \Rightarrow \bot \quad \text{if } t \ne f \land t \notin \EVar \\
		Orient & \insertUPpair P t x \Rightarrow \insertUPpair P x t \quad \text{if } t \notin \EVar \\
		Occurs check & \insertUPpair P x t \Rightarrow \bot \quad \text{if } x \in \FV{t} \\
		Elimination & \insertUPpair P x t \Rightarrow \insertUPpair{\subst P x t}{x}{t} \quad \text{if } x \notin \FV{t} \\
	\end{tabular}
	\caption{The new rules for the unification algorithm.}
	\label{fig:apprules}
\end{table}

In the case of \textbf{Decomposition}, it is now necessary to tackle functions one parameter at a time. We also need to add both sides to the unification problem, as it is no longer guaranteed that the left side is a function symbol, and it may contain further applications. Note that it is no longer possible to tell if the function symbols and their arity match until we fully decompose the application, and the symbols themselves will also be added to the set.

This is not a problem, however, because if the symbols and the arity did match, then the \textbf{Delete} rule can be used to get rid of the extra pair. If they did not match, we will end up with a symbol and either a different symbol or an application in the set. These are the two cases that the new \textbf{Symbol clash} rules solve. We have decided to split this rule because it is easier to state. They no longer deal with function applications, instead they are specialized to be used at the end of a chain of \textbf{Decomposition}s, filtering out cases that the original \textbf{Symbol clash} would have done in a single step. The $t \ne f$ condition plays a dual purpose here. In case the arities of the function symbols matched, but the symbols themselves did not, $t$ will necessarily be some symbol $g$, that is different from $f$, satisfying the condition. If the arities did not match however, then $t$ will be an application, which obviously will not be equal to a symbol, therefore satisfying the condition again. However, $t$ may still be some variable $x$, and we must take care not to reject this symbol-variable pair using the new rules as those are not erroneous.

We demonstrate in \Cref{fig:comp} with some small examples how the new rules compare to the old ones presented and explained in \Cref{sec:related}.

\newcommand{\UPone}[2]{\singletonUP{#1}{#2}}
\newcommand{\UPtwo}[4]{\insertUPpair{\UPone{#1}{#2}}{#3}{#4}}
\newcommand{\UPthree}[6]{\insertUPpair{\UPtwo{#1}{#2}{#3}{#4}}{#5}{#6}}
\newcommand{\UPfour}[8]{\insertUPpair{\UPthree{#1}{#2}{#3}{#4}{#5}{#6}}{#7}{#8}}

\begin{figure}[htb]
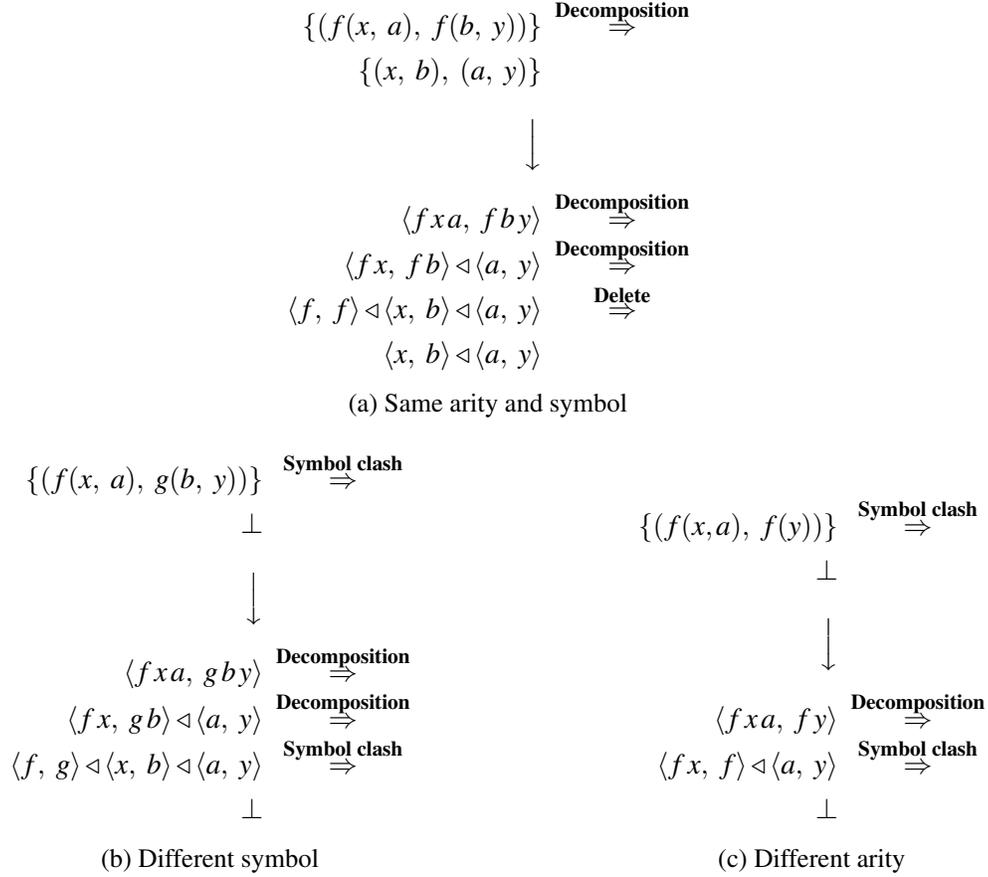

	\centering
	\begin{subfigure}{.90\textwidth}
		\centering
		\begin{tabular}{\ruletablespec}
			\{(f(x,\ a),\ f(b,\ y))\} & \rulearrow{Decomposition} \\
			\{(x,\ b),\ (a,\ y)\} & \rulephantom \\[1em]
			\Big\downarrow & \\[1em]
			\UPone{f \, x \, a}{f \, b \, y} & \rulearrow{Decomposition} \\
			\UPtwo{f \, x}{f \, b}{a}{y} & \rulearrow{Decomposition} \\
			\UPthree{f}{f}{x}{b}{a}{y} & \rulearrow{Delete} \\
			\UPtwo{x}{b}{a}{y} & \rulephantom \\
		\end{tabular}

		\caption{Same arity and symbol}
	\end{subfigure}

	\vspace*{1em}
	\hfill
	\begin{subfigure}{.40\textwidth}
		\centering
		\begin{tabular}{\ruletablespec}
			\{(f(x,\ a),\ g(b,\ y))\} & \rulearrow{Symbol clash} \\
			\bot & \rulephantom \\[1em]
			\Big\downarrow & \\[1em]
			\UPone{f \, x \, a}{g \, b \, y} & \rulearrow{Decomposition} \\
			\UPtwo{f \, x}{g \, b}{a}{y} & \rulearrow{Decomposition} \\
			\UPthree{f}{g}{x}{b}{a}{y} & \rulearrow{Symbol clash} \\
			\bot & \rulephantom \\
		\end{tabular}

		\caption{Different symbol}
	\end{subfigure}
	\hfill
	\begin{subfigure}{.45\textwidth}
		\centering
		\begin{tabular}{\ruletablespec}
			\{(f(x, a),\ f(y))\} & \rulearrow{Symbol clash} \\
			\bot & \rulephantom \\[1em]
			\Big\downarrow & \\[1em]
			\UPone{f \, x \, a}{f \, y} & \rulearrow{Decomposition} \\
			\UPtwo{f \, x}{f}{a}{y} & \rulearrow{Symbol clash} \\
			\bot & \rulephantom \\
		\end{tabular}

		\caption{Different arity}
	\end{subfigure}
	\hfill

	\caption{Comparison of old and new rules. Notice how $(a,\ y)$ and $(x,\ b)$ are not filtered out by either version of \textbf{Symbol clash} as they are not errors.}
	\label{fig:comp}
\end{figure}

\subsection{Soundness}

In this section we present supporting lemmas similar to those discussed in~\cite{unificationpaper}; however, we emphasise that they prove these results semantically, while we construct syntactical proofs. At the end of the section, we prove the soundness of the rule-based unification.

\begin{lemma}\label{prop3}
	Let $\varphi_1$ and $\varphi_2$ be functional patterns. Then $\Gamma
	\vdash \varphi_1 \land \varphi_2 \leftrightarrow \varphi_1 \land
	(\varphi_1 = \varphi_2)$.
\end{lemma}

\begin{proof}
First we prove the $\leftarrow$ direction. Here we simply split the premise and rewrite the goal using \ref{rule:rewritelocal} with $\varphi_1 = \varphi_2$ to obtain $\varphi_1 \land \varphi_1$, which is easily solved by the $\varphi_1$	hypothesis.

For the $\rightarrow$ direction, we know from \Cref{phiimpldefphi}that $\varphi \rightarrow \defined \varphi$. We apply this to the	$\varphi_1 \land \varphi_2$ hypothesis, that gives us a defined conjunction, which is the definition of $\in$, so our new hypothesis is $\varphi_1 \in \varphi_2$.

Then, we discharge $\varphi_1$ by using the left hand side of the original hypothesis, and $\varphi_1 \in \varphi_2$ implies $\varphi_1 = \varphi_2$ as shown in \Cref{membimplequal}, so we can solve that part as well.
\end{proof}

This lemma of utmost practical relevance, since it is used to set up the unification problem by introducing the equality pattern $\varphi_1 = \varphi_2$ which can be turned into the unification problem $\singletonUP{\varphi_1}{\varphi_2}$ to solve.

\begin{lemma} \label{lemma1}
	Let $\varphi$ and $t$ be patterns, and $x$ be a variable. Then $\Gamma
	\vdash x = t \rightarrow \subst{\varphi}{x}{t} = \varphi$.
\end{lemma}

\begin{proof}
    For technical reasons, we do not use the rewrite rule here, but the congruence of equality~\cite{matchingmulogic} on $x = t$,
	using the pattern context $\subst{\varphi}{x}{\square} = \varphi$. We can assume
	that $\square$ is a fresh variable in $\varphi$. This gives us
	$\subst{\subst{\varphi}{x}{\square}}{\square}{x} = \varphi \leftrightarrow
	\subst{\subst{\varphi}{x}{\square}}{\square}{t} = \varphi$.

	In $\subst{\subst{\varphi}{x}{\square}}{\square}{x}$ we replace all occurrences of $x$
	with a fresh variable $\square$, and then immediately replace all
	occurrences of that with $x$. It is clear to see that this operation is
	the same as $\subst{\varphi}{x}{x}$, which actually does nothing and so the
	pattern is equivalent to just $\varphi$.

	We can make a similar argument in the case of
	$\varphi[\square/x][t/\square]$, however, this time it is $t$ that we
	place back into $\square$, so this pattern is equivalent to
	$\varphi[t/x]$.

	Rewriting the formula with these two equivalences, we get $\varphi =
	\varphi \leftrightarrow \varphi[t/x] = \varphi$.

	As equivalence is just the conjunction of two implications, we can drop
	the $\leftarrow$ direction, keeping only the $\rightarrow$ one.

	Finally, $\varphi = \varphi$ is proven by reflexivity, so by modus
	ponens on the implication, we arrive at our goal of $\varphi[t/x] =
	\varphi$.
\end{proof}

The above lemma makes it easier to simplify substitutions using congruence
in future proofs.

\begin{lemma} \label{lemma2}
	Let $\varphi$ be a pattern and $\sigma$ a substitution. Then $\Gamma
	\vdash \varphi\sigma \land \phi^\sigma \leftrightarrow \varphi \land
	\phi^\sigma$.
\end{lemma}

\begin{proof}
	First we extract the $\phi^\sigma$ from both sides of the equivalence
	using \Cref{ecfequiv}, which yields us $\phi^\sigma \rightarrow
	\varphi\sigma \leftrightarrow \varphi$.

	Then we do induction on $\sigma$ with any $\varphi$. If it is the empty
	substitution, we can prove $\top \rightarrow \varphi \leftrightarrow
	\varphi$ by reflexivity. If it is not, then there is at least one $x$
	element variable and $t$ pattern that is part of $\sigma$ ($\sigma = \{x
	\mapsto t\} \cup \sigma'$), and we also know that $x = t$ is part of
	$\phi^\sigma$ ($\phi^\sigma = (x = t) \land \phi^{\sigma'}$). Thus our
	goal becomes $x = t \land \phi^{\sigma'} \rightarrow
	\varphi[t/x]\sigma' \leftrightarrow \varphi$.

	If we specialize the induction hypothesis with $\varphi[t/x]$ as
	$\varphi$, we get $\phi^{\sigma'} \rightarrow \varphi[t/x]{\sigma'}
	\leftrightarrow \varphi[t/x]$.

	$\phi^{\sigma'}$ is part of our hypothesis. Using the transitivity and
	symmetry of equivalence with this and the goal, we can reduce the goal to $\varphi[t/x]
	\leftrightarrow \varphi$.

	We still have $x = t$ in the hypothesis, and we can use \Cref{lemma1} on it to obtain $\varphi[t/x] = \varphi$.

	Finally, using \ref{rule:dt} we can extract the equivalence	from this equality, proving the goal.
\end{proof}

This is, in a sense, an extension of \Cref{lemma1} to set-based
substitutions, with the added requirement for the predicate $\phi^\sigma$.

Now we move on to justifying the soundness of single unification steps, and then argue about the correctness of unification step sequences by using induction.

\begin{lemma} \label{lemma3}
	Let $P$ and $P'$ be unification problems. If $P \Rightarrow P'$ and $P'
	\ne \bot$ then $\Gamma \vdash \toPredUP P \rightarrow \toPredUP{P'}$.
\end{lemma}

\begin{proof}
	We do induction on the $\Rightarrow$ relation. In most cases we will
	use \Cref{toPredicateInsertUP}, which states for
	every $P$ unification problem and $x$, $y$ terms that $\toPredUP{\insertUPpair P x y} \leftrightarrow x = y \land \toPredUP P$.

\begin{itemize}

\item
	Using the property in the \textbf{Delete} case, we get $t = t \land
	\toPredUP P \rightarrow \toPredUP P$ which is trivially proven.

\item
	In the case of \textbf{Decomposition}, we use the property three
	times, ultimately transforming the goal to $\app f t = \app g u \land
	\toPredUP P \rightarrow t = u \land f = g \land \toPredUP P$. This is the one
	time where we need to make use of the injectivity axiom, applying it to
	$\app f t = \app g u$ we obtain $f = g \land t = u$. With this now we can
	easily prove the goal.

\item
	With \textbf{Symbol clash L}, \textbf{Symbol clash R}, and later
	\textbf{Occurs check} we actually have a contradiction, as in both
	cases $P' = \bot$ and we have a hypothesis that says $P' \ne \bot$,
	therefore these cases are immediately discharged.

\item
	In the \textbf{Orient} case we use the property twice to transform the
	goal to $x = t \land \toPredUP P \rightarrow t = x \land \toPredUP P$ and this is
	solved using the symmetry of equality.

\item
	Finally, in the case of \textbf{Elimination} we will once again use the
	property twice, however, this time that will not be enough. We can
	reduce the goal to $x = t \land \toPredUP P \rightarrow x = t \land
	\toPredUP{\subst P t x}$ and the $x = t$ part is solvable like before,
	but for the other one we need to take advantage of another property,
	\Cref{toPredicateSubstituteAllUP}. This says that $\toPredUP{\subst P t x} \leftrightarrow \subst{\toPredUP P}{x}{t}$. Rewriting with this, we get $x = t
	\land \toPredUP P \rightarrow x = t \land \subst{\toPredUP P}{x}{t}$. If we choose
	$\sigma$ to be $\{x \mapsto t\}$ and $\varphi$ to be $\toPredUP P$, then this
	is a direct consequence of \Cref{lemma2}.

\end{itemize}
\end{proof}

We can show that not only single steps but their sequence is correct:

\begin{lemma} \label{lemma4}
	If $P$ and $P'$ are unification problems, and $P
	\Rightarrow^* P'$, where $P' \ne \bot$, then $\Gamma \vdash \toPredUP P
	\rightarrow \toPredUP{P'}$.
\end{lemma}

\begin{proof}
	First we do induction on the $\Rightarrow^*$ relation.

	In the reflexive case, we have $\toPredUP P \rightarrow \toPredUP P$, which is
	trivially provable.

	In the transitive case, we know that there is a $P''$ such that $P \Rightarrow P''$ and $P'' \Rightarrow^* P'$. We also have an inductive hypothesis that if $P' \ne \bot$, then $\toPredUP{P''} \rightarrow \toPredUP{P'}$ and from the hypothesis
	we know that $P' \ne \bot$. We need to use these facts to prove that $\toPredUP P
	\rightarrow \toPredUP{P''}$.

	For this we are going to use \Cref{lemma3} with $P \Rightarrow P''$,
	but in order to do that, we first need to prove that $P'' \ne \bot$.

	For that we are going to do another induction, this time on $P''
	\Rightarrow^* P'$. In the reflexive case, when $P'' = P'$, we can once
	again use the $P' \ne \bot$ hypothesis to do this. In the transitive case we get
	that $P'' \Rightarrow P'''$, and we know that in every case of
	$\Rightarrow$, the left side is an insertion into some unification
	problem that is not $\bot$. That is, there is some $P_0$, such that
	$P_0 \ne \bot$ and $P'' = \insertUPpair{P_0}{\_}{\_}$. But thanks to
	\Cref{insertNotBottomUP}, that said $P_0
	\ne \bot \rightarrow \insertUPpair{P_0}{\_}{\_} \ne \bot$, we know that $P''$
	cannot be $\bot$.

	With this condition satisfied, we can use \Cref{lemma3} to finish
	the proof.
\end{proof}

\begin{corollary} \label{lemma4coro}
	If a substitution $\sigma$ is the most general unifier of two patterns
	$t_1$ and $t_2$, then $\Gamma \vdash t_1 = t_2 \rightarrow
	\phi^\sigma$.
\end{corollary}

\begin{proof}
	Here we can use \Cref{thm:convenient} with the most general unifier
	$\sigma$ to obtain a $P'$ unification problem that is not $\bot$ and
	$\phi^{P'} = \phi^\sigma$. We also know that $t_1 = t_2$ is the same as
	$\phi^{\UPone{t_1}{t_2}}$, from \Cref{toPredicateSingletonUP}. By
	rewriting with these equalities, we have transformed the goal to be
	solvable by \Cref{lemma4}.
\end{proof}

\Cref{lemma4} extends \Cref{lemma3} to the
reflexive-transitive closure of the relation, while \Cref{lemma4coro}
specializes it to the singleton pattern that the algorithm starts with.

\begin{lemma} \label{lemma5}
	Let $\sigma$ be a unifier of two terms $t_1$ and $t_2$. Then $\Gamma
	\vdash \phi^\sigma \rightarrow t_1 = t_2$.
\end{lemma}

\begin{proof}
	First we take the $\phi^\sigma$ condition and add it to both sides of
	the equality using \Cref{ecfequal}, resulting in $t_1 \land
	\phi^\sigma = t_2 \land \phi^\sigma$.

	Next, we rewrite with \Cref{lemma2} on both sides of the equality, which
	gives us $t_1\sigma \land \phi^\sigma = t_2\sigma \land \phi^\sigma$.

	Now, we can remove the $\phi^\sigma$ from both sides by using
	\Cref{ecfequal} again, and the resulting $t_1\sigma = t_2\sigma$ is
	exactly the definition of $\sigma$ being a unifier of $t_1$ and $t_2$,
	which is in our hypothesis.
\end{proof}

This is the same as the above, but in the other direction.

\begin{theorem}[Soundness] \label{soundness}
	If a substitution $\sigma$ is the most general unifier of two patterns
	$t_1$ and $t_2$, then $\Gamma \vdash t_1 = t_2 \leftrightarrow
	\phi^\sigma$.
\end{theorem}

\begin{proof}
	This is the conjunction of \Cref{lemma4coro} and \Cref{lemma5}.
\end{proof}

The soundness theorem derives a result that allows us to manipulate conjunction patterns so that the conjunction of two structural patterns can be turned into the conjunction of a single structural pattern and a single predicate pattern. This is of practical importance as it allows matching logic provers to extract predicates to be discharged by external solvers:

\begin{corollary} \label{extractpredicate}
	Let $\sigma$ be the most general unifier of two patterns $t_1$ and
	$t_2$. Then $\Gamma \vdash t_1 \land t_2 \leftrightarrow t_1 \land
	\phi^\sigma$ and $\Gamma \vdash t_1 \land t_2 \leftrightarrow t_2 \land
	\phi^\sigma$.
\end{corollary}

\begin{proof}
	From \Cref{prop3} we know that $t_1 \land t_2 \leftrightarrow t_1
	\land t_1 = t_2$. Using the commutativity of $\land$ and symmetry of
	$=$, we can also get $t_1 \land t_2 \leftrightarrow t_2 \land t_1 =
	t_2$ from the same lemma. 

	We also know from \Cref{soundness} that $t_1 = t_2 \leftrightarrow
	\phi^\sigma$. Rewriting the previous two statements with this we get
	the two statements of this theorem.
\end{proof}

Finally, we demonstrate the algorithm and its soundness on an example.

\begin{example}
	We present the example shown in \Cref{polyex}, generate its unifier
	and prove the first statement of \Cref{extractpredicate} with its terms
	(the second may be proven similarly): There exists a $\sigma$ substitution such that $\Gamma \vdash
	\underbrace{f \, x \, (g \, 1) \, (g \, z)}_{t_1} \land \underbrace{f
	\, (g \, y) \, (g \, y) \, (g \, (g \, x))}_{t_2} \leftrightarrow
	\underbrace{f \, x \, (g \, 1) \, (g \, z)}_{t_1} \land \phi^\sigma$.

	We use existential quantification on $\sigma$ because we know that the terms are unifiable (therefore it does exist)
	and it can be automatically inferred as we run the algorithm.
\end{example}

\begin{proof}
	We begin with the $\rightarrow$ direction. First we use \Cref{prop3}
	to rewrite $t_1 \land t_2$ as $t_1 \land t_1 = t_2$.

	Then, we use \Cref{lemma4}, since we do not have a proof that
	$\sigma$ is the most general unifier (since we do not even know what
	$\sigma$ is yet), we cannot use \Cref{lemma4coro}. For this, we need
	to prove that there is some non-$\bot$ $P$ unification problem, such
	that $\UPone{t_1}{t_2} \Rightarrow^* P$. We will also not state what $P$
	is ahead of time, as that too can be inferred.

	This proof follows the steps outlined in \Cref{polyex}, but this
	time we need to use the applicative version of the
	\textbf{Decomposition} rule. We use the transitive constructor of
	$\Rightarrow^*$ with each step, and when we reach the solved form, end
	with the reflexive constructor. The steps are as follows:

	\begin{center}
	\begin{tabular}{\ruletablespec}
		\UPone{\highlight{f \, x \, (g \, 1) \, (g \, z)}}
			{\highlight{f \, (g \, y) \, (g \, y) \, (g \, (g \, x))}}
		 & \rulearrow{Decomposition} \\
		 \UPtwo{\highlight{f \, x \, (g \, 1)}}{\highlight{f \, (g \, y) \, (g \, y)}}{g \, z}{g \, (g \, x)}
		 & \rulearrow{Decomposition} \\
		 \UPthree{\highlight{f \, x}}{\highlight{f \, (g \, y)}}{g \, 1}{g \, y}{g \, z}{g \, (g \, x)}
		 & \rulearrow{Decomposition} \\
		 \UPfour{\highlight f}{\highlight f}{x}{g \, y}{g \, 1}{g \, y}{g \, z}{g \, (g \, x)}
		 & \rulearrow{Delete} \\
		 \UPthree{x}{g \, y}{\highlight{g \, 1}}{\highlight{g \, y}}{g \, z}{g \, (g \, x)}
		 & \rulearrow{Decomposition} \\
		 \UPfour{x}{g \, y}{\highlight g}{\highlight g}{1}{y}{g \, z}{g \, (g \, x)}
		 & \rulearrow{Delete} \\
		 \UPthree{x}{g \, y}{1}{y}{\highlight{g \, z}}{\highlight{g \, (g \, x)}}
		 & \rulearrow{Decomposition} \\
		 \UPfour{x}{g \, y}{1}{y}{\highlight g}{\highlight g}{z}{g \, x}
		 & \rulearrow{Delete} \\
		 \UPthree{x}{g \, y}{\highlight 1}{\highlight y}{z}{g \, x}
		 & \rulearrow{Orient} \\
		 \UPthree{x}{g \, y}{y}{1}{z}{g \, x}
		 & \rulephantom \\
	\end{tabular}
	\end{center}

	\vspace*{5mm}

	We could have used \textbf{Elimination} at the end, however, it will
	not help us in this proof, so we omitted those steps of the original
	example.
	Now that we have successfully inferred what $P$ is, we can easily prove
	that it is not $\bot$ and we have a hypothesis that says $t_1 = t_2
	\rightarrow x = g \, y \land y = 1 \land z = g \, x$. We can
	apply this to our hypothesis of $t_1 = t_2$.
	With this, we can actually already solve our goal. The $t_1$ part is
	trivial, and the $\phi^\sigma$ part is unknown, because $\sigma$ is
	uninstantiated, however, because it is a predicate, we know that it is a
	chain of conjunctions. If we solve this with the one
	from above, it will be inferred that $\sigma$ is indeed $\{(x \mapsto g \,
	y),\ (y \mapsto 1),\ (z \mapsto g \, x)\}$. 
	Now that we know this, the $\leftarrow$ direction can be solved. Here,
	the now known $\phi^\sigma$ gives us a series of equalities which we
	can use to repeatedly rewrite subterms in our goal of $t_1 \land t_2$
	and the $t_1$ hypothesis, until the $t_1$ and $t_2$ in them become the
	same pattern. From this the solution is trivial.

	We intuitively know that this rewriting is possible, as $\sigma$ is a
	unifier of $t_1$ and $t_2$, therefore $t_1\sigma = t_2\sigma$, and
	although we did not formally prove this, it is easy to manually check
	in this example. In this manual proof it is also not necessary to do
	the substitution fully (i.e. to calculate the actual $t_1\sigma$ and
	$t_2\sigma$), it is enough to go until the two terms match.
\end{proof}

\subsection{Mechanization}

As part of our work, we mechanized all the above presented results in the Coq proof assistant, based on an existing formalisation~\cite{bereczky2022mechanizing} of applicative matching logic. The formalisation uses a locally-nameless variable representation~\cite{locallynameless} and deep embedding, which makes the mathematical proofs and the formal proofs somewhat diverging here and there due to the additional technical complexity in the implementation stemming from  well-formedness constraints. However, the formalisation provides an implementation of the sequent calculus in terms of an embedded proof mode~\cite{proofmodepaper}, so our proof descriptions resemble the actual formal proof scripts~\cite{unifpr}.

\paragraph{Instances of the abstract unification problem.}

The abstract unification problem was implemented as a type class, and an instance of this class has been created for sets containing pairs of well-formed patterns, wrapped in an \texttt{option}. The optional type allows us to use the \texttt{None} value to represent failed unification problems. We are using stdpp's set classes~\cite{stdppbase}, therefore this instance is still generic in the exact implementation of the set, as long as it has all required methods and properties. In the future, we plan to instantiate this class for lists and other containers and investigate their non-functional properties.

\section{Conclusion}
\label{sec:conclusion}

This paper presented a generalisation of the well-known unification problem of symbolic expressions. We defined abstract unification problems for term patterns in the applicative variant of matching logic, and following in the footsteps of the related work we defined a rule-based algorithm for solving abstract unification problems. We instantiated abstract problems to sets and demonstrated their behaviour via examples. Last but not least, we proved the soundness of the unification algorithm, syntactically, using a sequent calculus for matching logic. The entire development is supported by a mechanisation of the formal theory, implemented in the Coq proof assistant.

\subsection*{Acknowledgements} We warmly thank Runtime Verification Inc. for their generous funding support.

\bibliographystyle{eptcs}
\bibliography{biblio}
\end{document}